\documentclass[letterpaper, 10 pt, conference]{ieeeconf}  % Comment this line out if you need a4paper

\IEEEoverridecommandlockouts                              % This command is only needed if 
\usepackage[table]{xcolor}
\usepackage{siunitx}
\usepackage{amsmath, amssymb, amsfonts}  % Math stuff
\usepackage[mathscr]{eucal}
\usepackage{cases,enumerate}
\usepackage{mathtools}
\usepackage{blkarray, bigstrut} %

\usepackage{bm}
\usepackage{bbm}
\usepackage{xcolor}
\usepackage{ifpdf}
\usepackage{mleftright}
\usepackage{multicol}
\usepackage{indentfirst}
\usepackage{soul}

\usepackage{graphics} % for pdf, bitmapped graphics files
\usepackage{graphicx}           % For including graphics
\graphicspath{{./Images/}}

\let\labelindent\relax
\usepackage{enumitem}  
\setlist[description]{leftmargin=\parindent}

\usepackage{subfigure}

\usepackage{ifpdf}

\usepackage{theorem}
\newtheorem{theorem}{Theorem}
\newtheorem{definition}{Definition}
\newtheorem{corollary}{Corollary}
\newtheorem{lemma}{Lemma}

\newtheorem{assumption}{Assumption}

{\theorembodyfont{\upshape}
	\newtheorem{remark}{Remark}
	\newtheorem{example}{Example}
	
}

\newcommand{\V}{\mathcal{V}}

\newcommand{\diag}{\mathrm{diag}}

\author{Lingfei Wang, Yu Xing, Karl H. Johansson% <-this % stops a space
	\thanks{* This work was supported by the Knut and Alice Wallenberg Foundation Wallenberg Scholar Grant, the Swedish Research Council Distinguished Professor Grant 2017-01078, and the Swedish Foundation for Strategic Research SUCCESS FUS21-0026.}
	\thanks{Lingfei Wang, Yu Xing and Karl H. Johansson are with Division of Decision and Control Systems, School of Electrical Engineering and Computer Science, KTH Royal Institute of Technology, and with Digital Futures, Stockholm, Sweden. Email:
		{\tt\small \{lingfei, yuxing2, kallej\}@kth.se}}%
}

\begin{document}
	\title{\LARGE \bf	
		% Concentration of the Friedkin-Johansen model on random networks \\
		% or Joint influence of network structure and stubbornness on the Friedkin-Johnson model\\
		% or Influence of network connectivity for opinion dynamics with partially stubborn agents\\
		% or 
		% Influence of stubborn community on final opinions of the Friedkin-Johnsen model\\
		On final opinions of the Friedkin-Johnsen model over random graphs with partially stubborn community}
	
	\maketitle
	
	\begin{abstract}
		%The randomness of social interaction can affect individuals' opinion forming process. On the other hand, some deterministic pattern may appear in the network structure as its size grows large. Taking both aspects into consideration, 
		This paper studies the formation of final opinions for the Friedkin-Johnsen (FJ) model with a community of partially stubborn agents. %influencers. 
		The underlying network of the FJ model is symmetric and generated from a random graph model, in which each link is added independently from a Bernoulli distribution. 
		It is shown that the final opinions of the FJ model will concentrate around those of an FJ model over the expected graph as the network size grows, on the condition that the stubborn agents are well connected to other agents.
		%For the opinion model, concentration is shown to occur for the final opinions, i.e., if the network size is large enough, the final opinions of the FJ model will concentrate on those over the expected graph with high probability. 
		Probability bounds are proposed for the distance between these two final opinion vectors, respectively for the cases where there exist non-stubborn agents or not.
		% , for the cases where a subgroup of agents or all of the agents are partially stubborn. 
		%
		Numerical experiments are provided to illustrate the theoretical findings.
		%To illustrate the theoretical results, some numerical examples are provided. 
		%
		The simulation shows that, in presence of non-stubborn agents, the link probability between the stubborn and the non-stubborn communities affect the distance between the two final opinion vectors significantly. 
		Additionally, if all agents are stubborn, the opinion distance decreases with the agent stubbornness.  
	\end{abstract}
	
	\section{Introduction}
	The prosperous development of social media brings massive opportunities to people for exchanging their opinions, which makes it more and more important to understand how these opinions evolve for different purposes, such as rumor control \cite{jiang2020dynamic}, market design \cite{jin2016understanding}, media competition \cite{quattrociocchi2014opinion}.  Social opinion dynamics is to study the evolution of opinions under the influence of individuals' mutual interactions, which can be described by a social network \cite{wasserman1994social,lewin1947frontiers}. Various dynamical models over social networks have been proposed, based on different cognitive mechanisms \cite{degroot1974reaching,hegselmann2002opinion,friedkin1990social,altafini2012consensus,ye2019influence}. Many of them explore the opinion formation process involving stubborn agents \cite{proskurnikov2017tutorial,yildiz2013binary,friedkin1990social,friedkin1999influence}. One type of stubborn agents, called totally stubborn, are those who never change their own opinions. An intuitive interpretation of totally stubborn agents are public opinion leaders such as media outlets, electoral candidates \cite{katz1957two}, while non-stubborn agents can be viewed as opinion followers. For classic opinion models such as the DeGroot model \cite{proskurnikov2017tutorial} and the voter's model \cite{yildiz2013binary}, the group opinions can be shaped by those stubborn opinions in the long run. The concept of stubbornness can be extended if the opinions of the stubborn agents are allowed to change. Such agents are called partially stubborn. 
	A widely studied model that involves partially stubborn agents is the Friedkin-Johnsen (FJ) model \cite{friedkin1990social}, in which each agent constantly takes the weighted average of their initial and their neighbors' opinions to update their current opinions, and the stubbornness is reflected by the weight of the initial opinion. Typically, an FJ model will have opinion cleavage that is commonly observed in reality \cite{proskurnikov2017tutorial,friedkin2015problem}, and mathematically, the opinions of all agents eventually converge to the convex hull of the initial opinions of the stubborn agents. 
	
	An important aspect for the study of opinion dynamics is to predict the final opinions of a group of interacting individuals \cite{friedkin2015problem}, see \cite{castro2017opinion} for an instance on recommender systems. However, a common challenge exists in finding a proper model to describe the structures of social interactions. For this problem, a solution is to use random graph models, which is a useful framework to understand network formation \cite{bollobas1998random,erdHos1960evolution}, especially for large scale networks \cite{strogatz2001exploring}. In particular, random graph models can exhibit concentration properties, that is, the realized graphs are close to their expectations in terms of the associated Laplacian matrices \cite{le2017concentration} and vertex degrees \cite{szymanski2005concentration}. Such concentration properties can be used to characterize the final opinions for opinion dynamics models over random graphs, i.e., the final opinions can be approximated by the opinions evolving over the expected topology of the underlying random graph model. However, even if some existing works study the opinion evolution over random graphs \cite{wu2004social,zehmakan2020opinion,zhang2014opinion,benczik2009opinion}, few of them focus on the concentration behavior of final opinions. In the recent work \cite{xing2024concentration}, the concentration of final opinions is reported for the gossip-based opinion dynamics over random graphs with totally stubborn agents. The concentration behavior is also observed for some extended FJ models, for which the final opinions are explicitly characterized for specific random graph models in \cite{andreou2024opinion,fraiman2024opinion}.   
	
	In this paper, we explore the final opinions of the FJ model over an undirected random graph, in which each link is added independently from a Bernoulli distribution. According to the positivity of stubbornness, the agents are partitioned into two communities, stubborn and non-stubborn, and the stubborn community has homogenous stubbornness. Different from \cite{xing2024concentration}, the stubborn opinions can be influenced by the non-stubborn ones, which makes the stubborn agents a more natural representation of the opinion leaders who need to listen to their followers and adjust their own opinions \cite{glass2021social,pan2022peer}. Moreover, it is inaccessible to apply the techniques in \cite{xing2024concentration} directly to the analysis in this paper (see Remark \ref{re:compare} for more details). The model in this paper is also different from that of \cite{andreou2024opinion,fraiman2024opinion}, in which the models do not consider coexistence of stubborn and non-stubborn agents.
	
	The contribution of this paper is twofold. At first, by means of matrix concentration inequalities, we provide probability bounds for the distance between the final opinions of the FJ model over the random graph and the counterpart over the expected graph, respectively for the cases that stubborn and non-stubborn agents co-exist (Theorem \ref{th:concen}) and that all the agents are stubborn (Theorem \ref{th:all_stubborn}). We reveal that if the network size is large enough, the final opinions of the FJ model over the random graph will concentrate around those of the FJ model over the expected graph, given that the stubborn agents are well-connected to the non-stubborn agents. Additionally, numerical examples are provided to a special two-communities stochastic block model (SBM). It is shown that, in presence of non-stubborn agents, the link probability between the stubborn and the non-stubborn communities affects the distance between the two final opinion vectors significantly (Example \ref{ex:bd_degree}). Moreover, if all agents are stubborn, the opinion distance decreases with the agent stubbornness (Example \ref{ex:bd_stub}). 
	
	The paper is organized as follows: Section \ref{sec:preliminary} gives some preliminary knowledge, Section \ref{sec:prob} formulates the problem, the main technical results are reported in Section \ref{sec:result}, and numerical examples are shown in Section \ref{sec:example}. 
	
	\noindent\textbf{Notation.} All vectors are real column vectors and are denoted with bold lowercase letters $\bf{x},\bf{y},\dots$ The $i$-th entry of a vector $\mathbf{x}$ is denoted by $[\mathbf{x}]_i$ or, if no confusion arises, $x_i$. The symbol ${\rm diag}(\bf{x})$ represents the diagonal matrix with diagonal entries equal to the entries of $\bf{x}$. Matrices are denoted with the capital letters such as $A,B,\dots$, of entries $A_{ij}$ or $[A]_{ij}$. With some abuse of notation, for two square matrices $A, B$, $\diag(A,B)$ is the square matrix with $A, B$ as the two diagonal blocks and all the other entries being $0$. The identity matrix is denoted by $I_n$, with dimension sometimes omitted, depending on the context. The $n$-order vector and matrix with all entries being $0$ or $1$ are denoted by $\mathbf{0}_n$ or $\mathbf{1}_n$, respectively with the dimensions omitted if there is no confusion. Let $\mathbf{e}_i$ be the vector with the $i$th entry as $1$ and all the others as $0$. We use $[n]$ to represent the set $\{1,\dots,n\}$.
	Given a set $\mathcal{C}$, we use $|\mathcal{C}|$ to denote its cardinality. A square matrix $A$ is called (row) stochastic if $A\geq \mathbf{0}$ and $\mathbf{1}= A\mathbf{1}$. Let $\|\cdot\|$ be the operator of $2$-norm of vectors/matrices Given a real number $x$, let $[x]$ be the nearest integer to $x$. For two number sequences $f(n)$ and $g(n)$, we write $f(n)=O(g(n))$ if there exists a constant $C>0$ such that $|f(n)|<Cg(n)$ holds for all $n\in\mathbb{N}$. Given two real numbers $a$ and $b$, define $a\vee b, a\wedge b$ as $\max\{a, b\}$ and $\min\{a,b\}$, respectively.  
	
	%%%%%%%%%%%%%%%%%%%%%%%%%%%%%%%%%%%%%%%%%%%%%%%%%%%%%%%
	\section{Preliminaries}\label{sec:preliminary}
	
	\subsection{Graphs and random graph model}
	
	Consider a network with nodes (agents) indexed in $\mathcal{V}=[n]$. The network is represented by an undirected weighted graph $\mathcal{G}=(\mathcal{V},\mathcal{E},W)$, where $\mathcal{E}$ is a set of unordered pairs of nodes, and $\{i,j\}\in\mathcal{E}$ means that there is a link between node $i$ and node $j$. The matrix $W=[w_{ij}]\in\mathbb{R}_{\geq 0}^{n\times n}$ is a symmetric weight matrix, with $w_{ij}>0$ if and only if $\{i,j\}\in\mathcal{E}$. Moreover, if $W\in\{0,1\}^{n\times n}$, it is called the adjacency matrix of the graph $\mathcal{G}$. For each agent $i$, its neighbor set, denoted as $\mathcal{N}_i$, is the set of all the agents connected to $i$, i.e., $\mathcal{N}_i=\{j:\{i,j\}\in\mathcal{E}\}$. The degree of agent $i$, denoted as $d_i(\mathcal{G})$, is the total weights of the links between $i$ and its neighbors, i.e., $d_i(\mathcal{G})=\sum_{j\in\mathcal{N}_i}w_{ij}$. The degree matrix of $\mathcal{G}$ is defined as $D_\mathcal{G}:=\diag(d_1(\mathcal{G}),\dots,d_n(\mathcal{G}))$ (the argument $\mathcal{G}$ is sometimes omitted if the graph is clear from context). 
	
	In this paper, we consider symmetric random graphs with each edge generated independently from a Bernoulli distribution. Specifically, let $\Psi=[\psi_{ij}]\in[0,1]^{n\times n}$ be a symmetric probability matrix, the random graph model is then defined as follows. 
	
	\begin{definition}[Random graph model]
		A random graph model, denoted as $\mathrm{RG}(\V,\Psi)$, is a process that generates an undirected graph $\mathcal{G}=(\V,\mathcal{E},A)$, with each edge $\{i,j\}$ generated independently by the probability $\psi_{ij}$. The matrix $A=[a_{ij}]$ is the adjacency matrix of $\mathcal{G}$, that is, $\mathbb{P}(a_{ij}=1)=1-\mathbb{P}(a_{ij}=0)=\psi_{ij}$. The graph $\mathcal{G}$ is called a realization of $\mathrm{RG}(\V,\Psi)$. 
	\end{definition}
	
	Note that only the upper triangular part of $\Psi$ is used to generate the random links. A widely studied random graph model is the stochastic block model (SBM) \cite{abbe2018community,lee2019review}, defined as follows.
	
	\begin{definition}[SBM]
		Assume that the agent set $\V$ consists of $K\in\mathbb{N}_+$ disjoint communities, i.e., $\V=\cup_{k\in[K]}\V_k$ and $\V_{k_1}\cap\V_{k_2}=\emptyset, \forall k_1\not= k_2$. For each $i\in\V$, let $c_i$ be the index of the community to which the agent $i$ belongs, i.e., $i\in\V_{c_i}$. A random graph model $\mathrm{RG}(\V,\Psi)$ is called an SBM if there exists a symmetric matrix $\Pi=[\pi_{ij}]\in [0,1]^{K\times K}$, such that for any $i,j\in\V$, it holds $\psi_{ij}=\pi_{c_ic_j}, i\not= j$ and $\psi_{ii}=0$. 
	\end{definition}
	
	%%%%%%%%%%%%%%%%%%%%%%%%%%%%%%%%%%%%%%%%%%%%%%%
	\subsection{FJ model}
	
	The FJ model is an opinion dynamics model in which some agents behave stubbornly, in the sense  that they defend their positions while discussing with the other agents \cite{friedkin1999influence}. In the FJ model over a graph $\mathcal{G}=(\V,\mathcal{E},W)$ with $\V=[n]$, for each agent $i$, the update of its opinion is
	\begin{equation}
		x_i(t+1)=(1-\theta_i)\sum_{j\in\mathcal{N}_i}\frac{w_{ij}}{d_i}x_j(t)+\theta_i x_i(0), \quad t\in\mathbb{N},
	\end{equation}
	where $x_i(t)\in\mathbb{R}$ is the opinion of $i$ at time $t$, and $\theta_i\in[0,1]$ is the stubbornness of agent $i$. The compact form of the FJ model can be written as 
	\begin{equation}\label{eq:FJ}
		\mathbf{x}(t+1)=(I-\Theta)D_\mathcal{G}^{-1}W\mathbf{x}(t)+\Theta\mathbf{x}(0),\quad t\in\mathbb{N},
	\end{equation}
	where $\mathbf{x}(t):=[x_1(t),\dots,x_n(t)]^\top$ is the opinion vector and $\Theta:=\diag (\theta_1,\dots,\theta_n)$ is a diagonal stubbornness matrix; $D_\mathcal{G}^{-1}W$ is the normalized weight matrix of $\mathcal{G}$. 
	
	For the FJ model \eqref{eq:FJ} over the graph $\mathcal{G}=(\V,\mathcal{E},W)$, according to Theorem~21 of \cite{proskurnikov2017tutorial}, the following result holds. 
	
	\begin{lemma}\label{le:FJ_conv}
		For the FJ model \eqref{eq:FJ}, if for each agent $i$, either $\theta_i>0$ or $i$ is connected to some agent $j$ with $\theta_j>0$, the opinion vector $\mathbf{x}(t)$ will converge, and
		\begin{equation}\label{eq:form_P}
			\mathbf{x}(\infty):=\lim_{t\to\infty}\mathbf{x}(t)=P\mathbf{x}(0), \; P=(I-(I-\Theta)D_\mathcal{G}^{-1}W)^{-1}\Theta,
		\end{equation}
		where $P$ is a stochastic matrix.
	\end{lemma}

	%%%%%%%%%%%%%%%%%%%%%%%%%%%%%%%%%%%%%%%%%%%%%%%%
	\subsection{Concentration inequalities}
	
	To estimate the sum of independent random matrices, the following two concentration inequalities are usually used.
	
	\begin{lemma}[Bernstein inequality]
		Let $X_1,\dots, X_N\in\mathbb{R}^{n\times n}$ be independent zero-mean random matrices. Suppose that $\|X_i\|\leq U, i\in[N]$ holds with probability $1$. Then for any $a\geq 0$, it holds that 
		\begin{equation}\label{eq:bern_square}
			\mathbb{P}\left(\left\|\sum_{i=1}^NX_i\right\|\geq a\right) \leq 2n\exp\left\{\frac{-a^2/2}{\sigma^2+Ua/3}\right\},
		\end{equation}
		where $\sigma^2=\|\sum_{i=1}^N\mathbb{E}[X_i^2]\|$. 
		% If $X_1,\dots,X_N\in\mathbb{R}^{m\times n}$ are independent, mean zero, and such that $\|Y_i\|\leq K a.s.$. Then for any $a\geq 0$, it holds that 
		% \begin{equation}\label{eq:bern_non_square}
			%     \mathbb{P}\left(\|\sum_{i=1}^NX_i\|\geq a\right) \leq 2(m+n)\exp\left\{\frac{-a^2/2}{\sigma^2+Ka/3}\right\},
			% \end{equation}
		% where $\sigma^2=\max\left\{\|\sum_{i=1}^N\mathbb{E}[X_i^\top X_i]\|,\|\sum_{i=1}^N\mathbb{E}[X_iX_i^\top]\|\right\}$.
	\end{lemma}
	
	\begin{lemma}[Chernoff inequality]
		Suppose that $X_1, \dots, X_n$ are independent Bernoulli random variables such that $\mathbb{P}(X_i=1)=p_i=1-\mathbb{P}(X_i=0)$. Let $X=\sum_{i=1}^nX_i$ and $\mu=\mathbb{E}[X]=\sum_{i=1}^n p_i$. Then for $0<\delta<1$, it holds $\mathbb{P}(X\leq (1-\delta)\mu)\leq e^{-\mu\delta^2/2}$.
		% \begin{equation*}
			%     \begin{aligned}
				%         % \mathbb{P}(X\geq (1+\delta)\mu)&\leq e^{-\mu \delta^2/3},\\
				%         .
				%     \end{aligned}
			% \end{equation*}
	\end{lemma}
	
	\section{Problem formulation}\label{sec:prob}
	
	Consider a random graph model $\mathrm{RG}(\V,\Psi)$ with a group of agents indexed in $\V=[n]$. The probability matrix $\Psi=[\psi_{ij}]\in[0,1]^{n\times n}$ is symmetric. In this paper, we assume that no self-loop exists, i.e., $\psi_{ii}=0, \forall i\in\V$. Let $\mathcal{G}=(\V,\mathcal{E},A)$ be a realization of $\mathrm{RG}(\V,\Psi)$. Consider the FJ model \eqref{eq:FJ} over $\mathcal{G}$, with $W=A$. The opinion evolution is 
	\begin{equation}\label{eq:FJ_rand}
		\mathbf{x}(t+1)=(I-\Theta)D^{-1}A\mathbf{x}(t)+\Theta\mathbf{x}(0),
	\end{equation}
	where $\mathbf{x}(t)=[x_1(t),\dots,x_n(t)]^\top$ is the opinion vector, $D:=D_\mathcal{G}$ is the degree matrix of $\mathcal{G}$, and $\Theta=\diag(\theta_1,\dots, \theta_n)$ is the stubbornness matrix to be specified later on.

	In this paper, we assume that $n_s$ agents are (partially) stubborn, with fixed stubbornness $\theta\in(0,1)$, while all the other $n_r$ agents are non-stubborn. Note that $n_s+n_r=n$. We use $\V_s=\{1,\dots, n_s\}$ and $\V_r=\{n_s+1,\dots,n\}$ to represent the sets of stubborn agents and non-stubborn agents, respectively, that is
	\begin{equation}\label{eq:stb}
		\theta_i=\left\{\begin{array}{cc}
			\theta, & i\in\V_s, \\
			0, & i\in\V_r.
		\end{array}\right.
	\end{equation}
	In compact form, the stubbornness matrix $\Theta$ is 
	\[\Theta=\diag(\theta I_{n_s},\mathbf{0}_{n_r}). \]
	Accordingly, for the FJ model \eqref{eq:FJ_rand}, if the conditions for Lemma \ref{le:FJ_conv} are satisfied, it can be verified that
	\begin{equation}\label{eq:form_P}
		\begin{aligned}
			\mathbf{x}(\infty)=\lim_{t\to\infty}\mathbf{x}(t)=P\mathbf{x}(0), \quad       P=M^{-1}(I-\Theta)^{-1}\Theta D,
		\end{aligned}
	\end{equation}
	with
	\begin{equation}\label{eq:form_M}
		\begin{aligned}
			% M=(I-\Theta)^{-1}\Theta D+D-A.
			M=(I-\Theta)^{-1}\Theta D+D-A.
		\end{aligned}
	\end{equation}
	% Note that  
	% % \[(I-\Theta)^{-1}\Theta D=\left(\begin{array}{cc}
		% %            \frac{\theta}{1-\theta}D_{s}  & \mathbf{0} \\
		% %             \mathbf{0} & \mathbf{0}
		% %         \end{array} \right).\]
	% \[\Theta D=\diag(\theta D_{s},\mathbf{0}_{n_r}),\]
	% where the submatrix $D_{s}$ is the degree matrix of the stubborn agents. 
	
	Let $\Bar{\mathcal{G}}=(\V,\Bar{\mathcal{E}},\bar{A})$ be the expected graph of the random graph model $\mathrm{RG}(\V,\Psi)$, with $\Bar{A}:=\mathbb{E}[A]=\Psi$ as the expectation of the adjacency matrix $A$ and $\bar{\mathcal{E}}$ defined correspondingly.
	Under the division of stubborn/non-stubborn agents, i.e., $\V=\V_s\cup \V_r$, some parameters of vertex degrees of $\Bar{\mathcal{G}}$  are given as follows:
	
	\begin{itemize}
		\item $\delta^s:=\min_{i\in\V_s}d_i(\bar{\mathcal{G}})=\min_{i\in\V_s}\sum_{j\not=i}\psi_{ij}$: the minimum expected degree of stubborn agents;
		\item $\delta^{rs}:=\min_{i\in\V_r}\sum_{j\in\V_s}\psi_{ij}$: the minimum expected stubborn degree of non-stubborn agents;
		\item $\Delta^s:=\max_{i\in\V_s}d_i(\bar{\mathcal{G}})=\max_{i\in\V_s}\sum_{j\not=i}\psi_{ij}$: the maximum expected degree of stubborn agents;
		\item $\Delta^{rs}:=\max_{i\in\V_r}\sum_{j\in\V_s}\psi_{ij}$: the maximum expected stubborn degree of non-stubborn agents.
		\item $\Delta^r:=\max_{i\in\V_r}d_i(\bar{\mathcal{G}})=\max_{i\in\V_s}\sum_{j\not=i}\psi_{ij}$: the maximum expected degree of non-stubborn agents;
		\item $\Delta:=\Delta^r\vee \Delta^s$: the maximum expected degree of all the agents.
	\end{itemize}
	
	The following assumptions are used throughout the paper.
	
	\begin{assumption}\label{assum:bound}
		For the random graph model $\mathrm{RG}(\V,\Psi)$ with $\V=\V_s\cup \V_r$, there exists two constants $c_1, c_2>8$ such that $\delta^{rs}\geq c_1 \log n $ and $\delta^{s}\geq c_2 \log n$.
	\end{assumption}
	
	\begin{assumption}\label{assum:cvg}
		For the random graph model $\mathrm{RG}(\V,\Psi)$ with $\V=\V_s\cup \V_r$, given any agent $i$, if $i\in\V_r$, there must exist $j\in\V_s$ such that $\psi_{ij}>0$.
	\end{assumption}

	In \cite{xing2024concentration}, opinion concentration is proved for the gossip-based DeGroot model over randoms graphs with totally stubborn agent. That is, the final opinions tend to concentrate around those of the expected opinion dynamics if the network size is large. Motivated by this, in this paper, we want to characterize the concentration property of the final opinion vector $\mathbf{x}(\infty)$ in \eqref{eq:form_P}, under the randomness of the adjacency matrix $A$. 
	
	Specifically, consider the opinion dynamics \eqref{eq:FJ} over the expected graph $\Bar{\mathcal{G}}$ with the same initial opinions and stubbornness, that is 
	\begin{equation}\label{eq:FJ_exp}
		\Bar{\mathbf{x}}(t+1)=(I-\Theta)\Bar{D}^{-1}\Bar{A}\Bar{\mathbf{x}}(t)+\Theta\Bar{\mathbf{x}}(0).
	\end{equation}
	where $\Bar{\mathbf{x}}(t)$ represents the opinion vector for the expected dynamics at time $t$, with $\Bar{\mathbf{x}}(0)=\mathbf{x}(0)$; $\Bar{D}=D_{\Bar{\mathcal{G}}}$ is the degree matrix of $\Bar{\mathcal{G}}$. 
	Under Assumption \ref{assum:cvg}, from Lemma \ref{le:FJ_conv}, the opinion vector $\bar{\mathbf{x}}(t)$ converges, and 
	\begin{equation}\label{eq:form_P_exp}
		\bar{\mathbf{x}}(\infty):=\lim_{t\to\infty}\bar{\mathbf{x}}(t)=\bar{P}\mathbf{x}(0), \; \Bar{P}=\Bar{M}^{-1}(I-\Theta)^{-1}\Theta \Bar{D},
	\end{equation}
	where $\Bar{M}$ is defined the same as in \eqref{eq:form_M}, with the matrices $D, A$ replaced by $\Bar{D}$ and $\Bar{A}$. The problem of interest is stated as follows: \newline
	
	\noindent\textbf{Problem}. For the opinion model \eqref{eq:FJ_rand} with partially stubborn community, this paper is going to provide some probability bounds to the opinion distance $\|\mathbf{x}(\infty)-\Bar{\mathbf{x}}(\infty)\|$, with relation to the network size $n$. Moreover, we also want to understand how these bounds are affected by the parameters of the random graph model $\mathrm{RG}(\V,\Psi)$, as well as the agent stubbornness $\theta$.

	%%%%%%%%%%%%%%%%%%%%%%%%%%%%%%%%%%%%%%%%%%%%
	\section{Main results}\label{sec:result}
	
	This section gives probability bounds for the opinion distance $\|\mathbf{x}(\infty)-\Bar{\mathbf{x}}(\infty)\|$. To estimate $\|\mathbf{x}(\infty)-\Bar{\mathbf{x}}(\infty)\|$, a probability lower bound for the minimum eigenvalue of the matrix $M$ (defined in \eqref{eq:form_M}) needs to be given beforehand.

	\subsection{Lower bound for $\lambda_{\min}(M)$}
	
	Before going to the main results, some definitions regarding link weights of the random graph $\mathcal{G}$ need to be given. 
	
	\begin{itemize}
		\item $d_i^s:=|\{j\in\mathcal{N}_i\cap \V_s\}|, d_i^r:=|\{j\in\mathcal{N}_i\cap \V_r\}|$: the stubborn and non-stubborn degree of agent $i$, respectively; 
		\item $d_{\min}^s=\min_{i\in\V_s}d_i(\mathcal{G})$: the minimum degree of stubborn agents, 
		\item $d_{\min}^{rs}=\min_{i\in\V_r}d_i^s$: the minimum stubborn degree of non-stubborn agents.
	\end{itemize}
	
	It is easy to see that $M$ is a symmetric matrix. According to the Gershgorin disk theorem, $M$ is positive semi-definite. In fact, the eigenvalues of $M$ can be lower bounded by the values defined above. 
	
	\begin{lemma}\label{le:min_eig}
		Consider the model \eqref{eq:FJ_rand} over the random graph $\mathcal{G}$, with agent stubbornness defined in \eqref{eq:stb}. Assume that each non-stubborn agent is connected to some stubborn one. Then, for the matrix $M$ defined in \eqref{eq:form_M}, it holds $\lambda_{\min}(M)\geq b_1$, with
		\begin{equation*}
			\begin{aligned}
				b_1=
				\frac{\theta d_{\min}^sd_{\min}^{rs} }{d_{\min}^s+(1-\theta)d_{\min}^{rs}}.
			\end{aligned}
		\end{equation*}
	\end{lemma}
	
	\begin{proof}
		For the simplicity of notation, let $\lambda:=\lambda_{\min}(M)\geq 0$. Due to that $M$ is a real symmetric matrix, it must have a real eigenvector corresponding to the real eigenvalue $\lambda$, denoted as $\mathbf{v}$. Let $v_{\max}^r:=\max_{i\in\V_r} |v_i|$ and $v_{\max}^s:=\max_{i\in\V_s} |v_i|$. According to the equation $M\mathbf{v}=\lambda\mathbf{v}$, for any $i\in\V_s$, it holds
		\begin{equation*}
			\begin{aligned}
				(\frac{d_i\theta}{1-\theta}+d_i-\lambda)v_i=\sum_{j\in\mathcal{N}_i}v_j.
			\end{aligned}
		\end{equation*}
		Suppose that $\lambda<\frac{\theta d_{\min}^s}{1-\theta}$. Taking absolute values of both sides and using the triangular inequality, we then have
		\begin{equation}
			\begin{aligned}
				(\frac{d_i\theta}{1-\theta}+d_i-\lambda)v_{\max}^s\leq d_i^rv_{\max}^r+d_i^sv_{\max}^s.
			\end{aligned}
		\end{equation}
		That is, $(\frac{d_i\theta}{1-\theta}+d_i^r-\lambda)v_{\max}^s\leq d_i^rv_{\max}^r$, 
		which implies that $v_{\max}^s<v_{\max}^r$. Furthermore, for all $i\in\V_s$, it holds
		\begin{equation}\label{ineq:stb_a}
			\begin{aligned}
				v_{\max}^s\leq \frac{d_i^r}{\frac{d_i\theta}{1-\theta}+d_i^r-\lambda}v_{\max}^r.
				% v_{\max}^s\leq \frac{(1-\theta)d_i^r}{\theta d_i+(1-\theta)d_i^r-\lambda}v_{\max}^r, \quad \forall i\in\V_s.
			\end{aligned}
		\end{equation}
		On the other hand, for any $i\in\V_r$, it holds
		\[(d_i-\lambda)v_i=\sum_{j\in\mathcal{N}_i}v_j. \]
		Following similar argument as above, we have
		\begin{equation*}
			\begin{aligned}
				|d_i-\lambda|v_{\max}^r&\leq d_i^rv_{\max}^r+d_i^sv_{\max}^s.
			\end{aligned}
		\end{equation*}
		Suppose that $\lambda<d_{\min}^{rs}$. We have $(d_i^s-\lambda)v_{\max}^r\leq d_i^sv_{\max}^s$,
		which gives that
		\begin{equation}\label{ineq:v_nstb}
			\begin{aligned}
				v_{\max}^r\leq \frac{d_{\min}^{rs}}{d_{\min}^{rs}-\lambda}v_{\max}^s.
			\end{aligned}
		\end{equation}
		Taking \eqref{ineq:v_nstb} into \eqref{ineq:stb_a}, we obtain
		\begin{equation}
			\begin{aligned}
				1\leq \frac{d_i^r}{\frac{\theta d_i}{1-\theta}+d_i^r-\lambda}\cdot \frac{d_{\min}^{rs}}{d_{\min}^{rs}-\lambda},
			\end{aligned}
		\end{equation}
		that is,
		\begin{equation}
			\begin{aligned}
				\lambda^2-(\frac{\theta}{1-\theta}d_i+d_i^r+d_{\min}^{rs})\lambda+\frac{\theta}{1-\theta}d_id_{\min}^{rs}\leq 0.
				% \lambda^2-(\theta d_i+(1-\theta) d_i^r+d_{\min}^{rs}))\lambda+\theta d_id_{\min}^{rs}\leq 0, \; \forall i\in\V_s.
			\end{aligned}
		\end{equation}
		Therefore, for all $i\in\V_s$, it must be
		\[-(\frac{\theta}{1-\theta}d_i+d_i+d_{\min}^{rs})\lambda+\frac{\theta}{1-\theta}d_id_{\min}^{rs}\leq 0, \]
		% \begin{equation*}
			%     \begin{aligned}
				%         &-(\theta d_i+(1-\theta) d_i^r+d_{\min}^{rs})\lambda+\theta d_id_{\min}^{rs}\\
				%         &= -(d_i-(1-\theta) d_i^s+d_{\min}^{rs})\lambda+\theta d_id_{\min}^{rs}  \leq 0, \; \forall i\in\V_s, 
				%     \end{aligned}
			% \end{equation*}
		which yields $ \lambda\geq \frac{\theta d_{\min}^sd_{\min}^{rs} }{d_{\min}^s+(1-\theta)d_{\min}^{rs}}$.
		% \begin{equation}
			%     \begin{aligned}
				%         \lambda\geq \frac{\theta d_{\min}^sd_{\min}^{rs} }{d_{\min}^s+(1-\theta)d_{\min}^{rs}}.
				%     \end{aligned}
			% \end{equation}
		Combining the arguments above, and due to the fact that 
		\begin{equation}
			\begin{aligned}
				&\min\left\{\frac{\theta d_{\min}^s}{1-\theta}, d_{\min}^{rs}, \frac{\theta d_{\min}^sd_{\min}^{rs} }{d_{\min}^s+(1-\theta)d_{\min}^{rs}} \right\}\\ 
				&= \frac{\theta d_{\min}^sd_{\min}^{rs} }{d_{\min}^s+(1-\theta)d_{\min}^{rs}},
			\end{aligned}
		\end{equation}
		the proof is then completed.
	\end{proof}
	
	\begin{remark}\label{re:compare}
		A similar problem is investigated in the proof of Theorem 4.3 in \cite{xing2023concentration}, in which the Gershgorin disk theorem is applied to estimate the minimum eigenvalue of a Laplacian-like matrix $\Bar{M}^{\mathcal{G}}$. However, for the matrix $M$ in \eqref{eq:form_M}, the lower bound of $\lambda_{\min}(M)$ given by the Gershgorin disk theorem is~$0$. In this sense, Lemma \eqref{le:min_eig} provides a more adjusted bound to $\lambda_{\min}(M)$.
	\end{remark}
	
	% Before going to the next lemma, some definitions need to be given at first. 
	% \begin{itemize}
		%     \item let $\delta^s:=\min_{i\in\V_s}\mathbb{E}[d_i]=\min_{i\in\V_s}\sum_{j\not=i}\psi_{ij}$ be the minimum expected degree of stubborn agents;
		%     \item let $\delta^r:=\min_{i\in\V_r}\mathbb{E}[d_i]=\min_{i\in\V_r}\sum_{j\not=i}\psi_{ij}$ be the minimum expected degree of non-stubborn agents;
		%     \item let $\delta^{rs}:=\min_{i\in\V_r}\mathbb{E}[d_i^s]=\min_{i\in\V_r}\sum_{j\not=i}\psi_{ij}$ be the minimum expected stubborn degree of non-stubborn agents;
		%     \item let $\Delta^{rr}:=\max_{i\in\V_r}\mathbb{E}[d_i^r]=\max_{i\in\V_r}\sum_{j\not=i}\psi_{ij}$ be the maximum expected regular degree of non-stubborn agents.
		% \end{itemize}
	
	Applying the Bernstein inequality to the random graph model $\mathrm{RG}(\V,\Psi)$, we can obtain the following lemma.
	
	\begin{lemma}\label{le:degree}
		Consider the random graph model $\mathrm{RG}(\V,\Psi)$. Let Assumption \ref{assum:cvg} hold. The following inequalities hold:
		\begin{equation*}
			\begin{aligned}
				&\mathbb{P}\left(d_{\min}^s\leq \frac{\delta^s}{2}\right)\leq n_se^{-\delta^s/8}, \,
				\mathbb{P}\left(d_{\min}^{rs}\leq  \frac{\delta^{rs}}{2}\right)\leq n_re^{-\delta^{rs}/8}.\\
				% &\mathbb{P}(d_{\min}^{rs}=0)\leq n_re^{-\delta^{rs}}.\\
				% &\mathbb{P}(d_{\min}^{ss}\leq \frac{\delta^{ss}}{2})\leq n_se^{-\delta^{ss}8}.
			\end{aligned}
		\end{equation*}
	\end{lemma}
	
	\begin{proof}
		Notice that for any $i\in\V_s$, it holds $d_i^s=\sum_{j\in\V_s}a_{ij}$. Due to that $a_{ij}, \forall j\in\V_s$ are independent Bernoulli variables, we can apply the Chernoff inequality and obtain $\mathbb{P}\left(d_i\leq \frac{\mathbb{E}[d_i]}{2}\right)\leq \exp\left\{-\frac{\mathbb{E}[d_i]}{8}\right\}$.
		Therefore, it further holds
		\begin{equation*}
			\begin{aligned}
				&\mathbb{P}(d_{\min}^s\leq \frac{\delta^s}{2})\leq \mathbb{P}(\cup_{i\in\V_s}\{d_i\leq \mathbb{E}[d_i]/2\})\\
				&\leq \sum_{i\in\V_s}\mathbb{P}(d_i\leq \mathbb{E}[d_i]/2)\leq \sum_{i\in\V_s}\exp\{-\frac{\mathbb{E}[d_i]}{8}\}\leq n_s e^{-\delta^s/8}.
			\end{aligned}
		\end{equation*}
		The first inequality of the lemma is then proved. Similarly, the second inequality can be proved by the same arguments. 
		% To prove the third inequality, notice that for any $i\in\V_r$, it holds
		% \begin{equation*}
			%     \begin{aligned}
				%         \mathbb{P}(d_i^s=0)=\Pi_{j\in\V}(1-\psi_{ij})\leq \exp\{-\sum_{j\in\V}\psi_{ij}\}\leq e^{-\delta^{rs}}. 
				%     \end{aligned}
			% \end{equation*}
		% Therefore, we have $\mathbb{P}(d_{\min}^{rs}=0)\leq \sum_{i\in\V_r}\mathbb{P}(d_i^s=0)\leq  n_re^{-\delta^{rs}}$. 
		The proof is then completed.
	\end{proof}
	
	Combing Lemmas \ref{le:min_eig} and \ref{le:degree}, a probability bound for $\lambda_{\min}(M)$ can be given in the following lemma.
	
	\begin{lemma}\label{le:es_M_inv}
		Consider the model \eqref{eq:FJ_rand} over the random graph $\mathcal{G}$, with agent stubbornness defined in \eqref{eq:stb}. Let Assumption \ref{assum:cvg} hold. It holds that 
		\begin{equation*}
			\begin{aligned}
				&\mathbb{P}\left(\lambda_{\min}(M)\geq \frac{\theta \delta^s\delta^{rs}/2}{\delta^s+(1-\theta)\delta^{rs}}\right)>1-\sigma_1,
			\end{aligned}
		\end{equation*}
		with $\sigma_1=n_se^{-\delta^s/8}+n_re^{-\delta^{rs}/8}$.
	\end{lemma}
	
	\begin{proof}
		From Lemma \ref{le:degree}, with probability $1-\sigma_1$, it holds $d_{\min}^s>\frac{\delta^s}{2}$ and $d_{\min}^{rs}>\frac{\delta^{rs}}{2}>0$. Therefore, the condition of Lemma \ref{le:min_eig} is satisfied. Noticing that $b_1$ in Lemma \ref{le:min_eig} is monotonically increasing with both $d_{\min}^s$ and $d_{\min}^{rs}$, the desired conclusion is then obtained. 
	\end{proof}
	
	\subsection{probability bound for $\|\mathbf{x}(\infty)-\Bar{\mathbf{x}}(\infty)\|$}
	
	By means of Lemma \ref{le:es_M_inv}, the following theorem can be obtained, giving a probability bound for $\|\mathbf{x}(\infty)-\Bar{\mathbf{x}}(\infty)\|$. 
	
	\begin{theorem}\label{th:concen}
		Consider the model \eqref{eq:FJ_rand} over the random graph $\mathcal{G}$, with agent stubbornness defined in \eqref{eq:stb}. Let Assumptions \ref{assum:bound} and \ref{assum:cvg} hold, and suppose 
		$\Delta^s\geq \log n$. It holds that 
		\[\mathbb{P}(\|\mathbf{x}(\infty)-\bar{\mathbf{x}}(\infty)\|\leq \epsilon_n\|\mathbf{x}(0)\|)> 1-\eta_n, \]
		where 
		\begin{subequations}\label{eq:prob_bd}
			\begin{align}
				\epsilon_n&= \frac{6}{1-\theta}\cdot\left(\frac{1}{\delta^{rs}}+\frac{1-\theta}{\delta^s}\right)\sqrt{\Delta^s \log n}+\frac{4(2-\theta)}{\theta(1-\theta)^2}\cdot \notag\\
				&\left(\frac{1}{\delta^{rs}}+\frac{1-\theta}{\delta^s}\right)^2\cdot
				\sqrt{\Delta\log n}\cdot \Delta^s,\label{eq:prob_bd1}\\
				\eta_n&= n_se^{-\frac{\delta^s}{8}}+n_re^{-\frac{\delta^{rs}}{8}}+2n_sn^{-\frac{3}{2}}
				+2n^{-\frac{1}{5}}=O(n^{-q}),\label{eq:prob_bd2}
			\end{align}
		\end{subequations}
		with $q=\min\{\frac{c_1\wedge c_2}{8}-1,\frac{1}{5} \}$.
	\end{theorem}
	
	\begin{proof}
		Under Assumption \ref{assum:cvg}, the matrix $\Bar{M}$ is invertible. From Lemma \ref{le:es_M_inv}, with probability at least of $1-\sigma_1$, $M$ becomes invertible and it holds 
		\begin{equation}\label{eq:es_M_inv}
			\|M^{-1}\|=\frac{1}{\lambda_{\min}(M)}<\frac{\delta^s+(1-\theta)\delta^{rs}}{\theta \delta^s\delta^{rs}/2}.
		\end{equation} 
		Therefore, according to Lemma \ref{le:FJ_conv}, it holds
		\begin{equation}\label{eq:es_dif}
			\begin{aligned}
				&\|\mathbf{x}(\infty)-\bar{\mathbf{x}}(\infty)\|=\|(P-\bar{P})\mathbf{x}(0)\|\\
				&=\|M^{-1}(I-\Theta)^{-1}\Theta D-\bar{M}^{-1}(I-\Theta)^{-1}\Theta\Bar{D}\|\|\mathbf{x}(0)\|\\
				&\leq \|M^{-1}(I-\Theta)^{-1}\Theta(D-\bar{D})\|\|\mathbf{x}(0)\|\\
				&+ \|(M^{-1}-\bar{M}^{-1})(I-\Theta)^{-1}\Theta\bar{D}\|\|\mathbf{x}(0)\|\\
				&= \|M^{-1}(I-\Theta)^{-1}\Theta\diag(D_s-\bar{D}_s,\mathbf{0})\|\|\mathbf{x}(0)\|\\
				&+\|(M^{-1}-\bar{M}^{-1})(I-\Theta)^{-1}\Theta\diag(\bar{D}_s,\mathbf{0})\|\|\mathbf{x}(0)\|,
			\end{aligned}
		\end{equation}
		where 
		\begin{equation*}
			\begin{aligned}
				D_s&=\diag(d_1(\mathcal{G}),\dots,d_{n_s}(\mathcal{G})),\\ 
				\Bar{D}_s&=\diag(d_1(\Bar{\mathcal{G}}),\dots,d_{n_s}(\Bar{\mathcal{G}}))
			\end{aligned}
		\end{equation*}
		are the degree matrices of the stubborn agents in the random graph $\mathcal{G}$ and expected graph $\Bar{\mathcal{G}}$, respectively, and the last equality is due to the associative law and that $\theta_i=0, \forall i\in\V_r$. Notice that $M^{-1}-\Bar{M}^{-1}=\Bar{M}^{-1}(\Bar{M}-M)M^{-1}$. We then have 
		\begin{equation}\label{eq:dif}
			\begin{aligned}
				\|\mathbf{x}(\infty)-\bar{\mathbf{x}}(\infty)\|&\leq \frac{\theta}{1-\theta}\|M^{-1}\|(\|D_s-\Bar{D}_s\|\\
				&+\|\Bar{M}^{-1}\|\|M-\Bar{M}\|\|\Bar{D}_s\|)\|\mathbf{x}(0)\|.
			\end{aligned}
		\end{equation}
		The following proof is divided into $3$ steps: at first, an estimation is given to $\|D_s-\Bar{D}_s\|$ and $\|M-\Bar{M}\|$, respectively, then the estimations are combined for a bound of $\|\mathbf{x}(\infty)-\bar{\mathbf{x}}(\infty)\|$. 
		
		\noindent \textbf{Step $1$: estimation of} $D_s-\bar{D}_s$. According to the Bernstein inequality, for any $a>0$, we have
		\begin{equation*}
			\begin{aligned}
				\mathbb{P}(\|D_s-\bar{D}_s\|\geq a) \leq 2n_s\mathrm{exp}\left\{\frac{-a^2}{2(2\Delta_{s}+a/3)} \right\}.
			\end{aligned}
		\end{equation*}
		Taking $a=3\sqrt{\Delta_s \log n}$, we have 
		\begin{equation}\label{eq:D_dif}
			\begin{aligned}
				\mathbb{P}(\|D_s-\bar{D}_s\|\geq a)&\leq 2n_s\exp\left\{\frac{-9\Delta_s \log n}{2(2\Delta_s+\sqrt{\Delta_s \log n})}\right\}\\
				&\leq 2n_s n^{-\frac{3}{2}},
			\end{aligned}
		\end{equation}
		where the last inequality is from $\Delta_s\geq \log n$. 
		
		\noindent \textbf{Step $2$: estimation of } $\|M-\Bar{M}\|$. Note that
		\[M-\bar{M}=\sum_{i=1}^{n-1}\sum_{j=i+1}^nX_{ji}, \]
		where $X_{ji}$ is defined as $X_{ji}:=(a_{ji}-\psi_{ji})E_{ji}$ such that
		\begin{equation*}
			\begin{aligned}
				&E_{ji}:=\\
				&\left\{\begin{array}{cl}
					\frac{1}{1-\theta}\left(\mathbf{e}_i\mathbf{e}_i^\top+\mathbf{e}_j\mathbf{e}_j^\top\right)-\mathbf{e}_i\mathbf{e}_j^\top-\mathbf{e}_j\mathbf{e}_i^\top,  & i,j\in\V_s, \\
					\frac{1}{1-\theta}\mathbf{e}_i\mathbf{e}_i^\top+\mathbf{e}_j\mathbf{e}_j^\top-\mathbf{e}_i\mathbf{e}_j^\top-\mathbf{e}_j\mathbf{e}_i^\top,  & i\in\V_s, j\in\V_r,\\
					\mathbf{e}_i\mathbf{e}_i^\top+\mathbf{e}_j\mathbf{e}_j^\top-\mathbf{e}_i\mathbf{e}_j^\top-\mathbf{e}_j\mathbf{e}_i^\top, & i,j\in\V_r.
				\end{array}\right.
			\end{aligned}
		\end{equation*}
		Accordingly, 
		\begin{itemize}
			\item for $i,j\in\V_s$, it holds
			\begin{equation*}
				\begin{aligned}
					\mathbb{E}[X_{ji}^2]&=(\psi_{ji}-\psi_{ji}^2)\left\{ \left(1+\frac{1}{(1-\theta)^2}\right)\left(\mathbf{e}_i\mathbf{e}_i^\top\right.\right.\\
					&\left.+\left.\mathbf{e}_j\mathbf{e}_j^\top\right)
					-\frac{2}{1-\theta}\left(\mathbf{e}_i\mathbf{e}_j^\top+\mathbf{e}_j\mathbf{e}_i^\top\right)\right\}.
				\end{aligned}
			\end{equation*}
			\item for $i\in\V_s,j\in\V_r$, it holds
			\begin{equation*}
				\begin{aligned}
					\mathbb{E}[X_{ji}^2]&=(\psi_{ji}-\psi_{ji}^2)\left\{ \left(1+\frac{1}{(1-\theta)^2}\right)\mathbf{e}_i\mathbf{e}_i^\top\right.\\
					&\left.+2\mathbf{e}_j\mathbf{e}_j^\top
					-(1+\frac{1}{1-\theta})\left(\mathbf{e}_i\mathbf{e}_j^\top+\mathbf{e}_j\mathbf{e}_i^\top\right)\right\}.
				\end{aligned}
			\end{equation*}
			\item for $i,j\in\V_r$, it holds
			\begin{equation*}
				\begin{aligned}
					\mathbb{E}[X_{ji}^2]&=2(\psi_{ji}-\psi_{ji}^2)\left( \mathbf{e}_i\mathbf{e}_i^\top+\mathbf{e}_j\mathbf{e}_j^\top
					-\mathbf{e}_i\mathbf{e}_j^\top-\mathbf{e}_j\mathbf{e}_i^\top\right).
				\end{aligned}
			\end{equation*}
		\end{itemize}
		Therefore, according to $\psi_{ji}-\psi_{ji}^2\leq \psi_{ji}$ and the Gershgorin disk theorem, it holds
		\begin{equation}
			\begin{aligned}
				&\|\sum_{i=1}^{n-1}\sum_{j=i+1}^n\mathbb{E}[X_{ji}^2]\|\leq \max\left\{ \left(1+\frac{1}{1-\theta}\right)^2\Delta^s\right.\\
				&\qquad\qquad-\frac{\theta}{1-\theta}\delta^{sr},
				\left.4\Delta^r+\frac{\theta}{1-\theta}\Delta^{rs}\right\}:=o^2\\
			\end{aligned}
		\end{equation}
		Applying the Gershgorin disk theorem again to the matrix $X_{ji}$, it can be verified that 
		\begin{equation*}
			\begin{aligned}
				\|X_{ji}\|=\lambda_{\max}(X_{ji})\leq
				1+\frac{1}{1-\theta}, \quad \forall j\in\V.
			\end{aligned}
		\end{equation*}
		According to the Bernstein theorem, for any  $a>0$, we have
		\begin{equation}
			\begin{aligned}
				\mathbb{P}(\|M-\bar{M}\|\geq a)\leq 2n\exp\left\{\frac{-a^2/2}{o^2+(1+\frac{1}{1-\theta})a/3}\right\}.
			\end{aligned}
		\end{equation}
		Taking $a=2(1+\frac{1}{1-\theta})\sqrt{\Delta\log n}$, it then holds
		\begin{equation}\label{eq:M_dif}
			\begin{aligned}
				\mathbb{P}(\|M-\bar{M}\|\geq a)\leq 2n\exp\left\{\frac{-2\log n}{1+\frac{2}{3}\sqrt{\frac{\log n}{\Delta}}}\right\}
				% &\left.\Big/ \left(1+\frac{2}{3}\sqrt{\frac{\log n}{\Delta^s\vee \Delta^{rs}}} \right)\right\}\\
				% &\leq 2(n+n_s)\exp\left\{-9(1+\frac{1}{1-\theta})\log n  \right.\\ 
				% &\left.\Big/ \left(2(1+\frac{1}{1-\theta})+2(\frac{\theta}{1-\theta}+2)\right)\right\}
				\leq 2n^{-\frac{1}{5}},
			\end{aligned}
		\end{equation}
		where the first inequality is due to that that $4\Delta^r+\frac{\theta}{1-\theta}\Delta^{rs}\leq (1+\frac{1}{1-\theta})^2\Delta^r$, and the second inequality is from $\theta<1$ and $\Delta\geq \Delta^s\geq \log n$.
		
		\noindent\textbf{Step $3$: overall estimation of $\|\mathbf{x}(\infty)-\bar{\mathbf{x}}(\infty)\|$}. Combining the \eqref{eq:es_M_inv}, \eqref{eq:D_dif} and \eqref{eq:M_dif}, we can conclude that with at least probability $1-\eta_n$, the following inequalities hold simultaneously 
		\begin{equation*}
			\begin{aligned}
				&\|M^{-1}\|<\frac{\delta^s+(1-\theta)\delta^{rs}}{\theta \delta^s\delta^{rs}/2}, \; \|D_s-\bar{D}_s\|<3\sqrt{\Delta_s \log n}, \; \\
				&\|M-\Bar{M}\|< 2(1+\frac{1}{1-\theta})\sqrt{\Delta\log n}.
			\end{aligned}
		\end{equation*}
		This, together with \eqref{eq:dif}, further gives that 
		\[\|\mathbf{x}(\infty)-\bar{\mathbf{x}}(\infty)\|\leq \epsilon_n\|\mathbf{x}(0)\|. \]
		Moreover, from Assumption \ref{assum:bound}, it holds 
		\begin{equation*}
			\begin{aligned}
				\eta_n&= n_se^{-\delta^s/8}+n_re^{-\delta^{rs}/8}+2n_sn^{-\frac{3}{2}}
				+2n^{-\frac{1}{5}}\\
				&\leq n^{-(\frac{c_2}{8}-1)}+n^{-(\frac{c_1}{8}-1)}+2n^{-\frac{1}{2}}+2n^{-\frac{1}{5}}=O(n^{-q}).
			\end{aligned}
		\end{equation*}
		The proof is then completed.
	\end{proof}
	
	% \begin{remark}
		%     According to the proof of Theorem \ref{th:concen}, $\epsilon_n$ is actually a probability bound for $\|P-\bar{P}\|$, i.e., 
		%     \[\mathbb{P}(\|P-\bar{P}\|\leq \epsilon_n)> 1-\eta_n, \]
		%     where $P$ and $\bar{P}$ are defined in \eqref{eq:form_P} and \eqref{eq:form_P_exp} respectively.
		% \end{remark}
	
	% With Proposition \ref{prop:concen}, the concentration property of the FJ model can be obtained directly in the following theorem. 
	
	% \begin{theorem}\label{th:concen}
		%     Consider the model \eqref{eq:FJ} on the random graph $\mathrm{RG}(\V, \Psi)$, with agent stubbornness defined in \eqref{eq:stb}. Let Assumptions \ref{assum:bound} and \ref{assum:cvg} hold. It holds that 
		%     \[\mathbb{P}(\|\mathbf{x}(\infty)-\bar{\mathbf{x}}(\infty)\|\leq \epsilon_n\|\mathbf{x}(0)\|)> 1-O(n^{-q}), \]
		%     where $\epsilon_n$ is defined as in Proposition \ref{prop:concen} and $q=\min\{\frac{c_1\wedge c_2}{8}-1,\frac{1}{5} \}$. 
		% \end{theorem}
	
	\begin{remark}
		According to \eqref{eq:prob_bd1}, if $\delta^{rs}\wedge \delta^s \geq cn$ for some $c>0$, which means that the expected graph $\Bar{\mathcal{G}}$ is well connected, then the opinion distance between $\mathbf{x}(\infty)$ and $\bar{\mathbf{x}}(\infty)$ (scaled by the norm of the initial opinion vector $\mathbf{x}(0)$) is close to~$0$ if the network size $n$ is large enough. In other words, Theorem~\ref{th:concen} implies that the final opinions of the FJ model \eqref{eq:FJ_rand} will concentrate around those of the expected FJ model \eqref{eq:FJ_exp}.
	\end{remark}
	
	\begin{remark}\label{re:bd_nstub_degree}
		From the form of $\epsilon_n$ in \eqref{eq:prob_bd}, it can be seen that the probability bound depends on the topology of the expected graph $\bar{\mathcal{G}}$ in a complex way. Nevertheless, as the coefficients in \eqref{eq:prob_bd} regarding network topology are all related to the link weights between the stubborn and non-stubborn communities, an indication is that the weights can affect the opinion distance $\|\mathbf{x}(\infty)-\Bar{\mathbf{x}}(\infty)\|$ significantly. This will be verified by the following Example \ref{ex:bd_degree}.
	\end{remark}
	
	It should be pointed out that the bound $\epsilon_n$ is conservative, especially when $\theta$ is close to $1$. This can be seen if we fix the network size $n$ and let $\theta$ converge to $1$, i.e., each agent tends to be fully stubborn and keep their opinions unchanged. Accordingly, the opinion distance $\|\mathbf{x}(\infty)-\bar{\mathbf{x}}(\infty)\|$ decreases to $0$. However, from $\eqref{eq:prob_bd}$, $\epsilon_n$ diverges to infinity if $\theta\to 1$. One reason for generating the conservative bound is that \eqref{eq:dif} may be an overestimate of $\|\mathbf{x}(\infty)-\bar{\mathbf{x}}(\infty)\|$ if there exist non-stubborn agents, since $\|M^{-1}\Theta\|=\theta\|M^{-1}|_{:,\V_s}\|$ ($M^{-1}|_{:,\V_s}$ is the submatrix of $M^{-1}$ consisting of the columns indexed in $\V_s$). Nevertheless, if all the agents are stubborn, the conservativeness can be reduced, as shown in the following theorem, in which we omit the superscript of $\delta^s$ and write it as~$\delta$. 
	
	\begin{theorem}\label{th:all_stubborn}
		Consider the model \eqref{eq:FJ_rand} over the random graph $\mathcal{G}$, with all agents having stubbornness $\theta$, i.e., $\V_s=\V$. Let Assumptions \ref{assum:bound} and \ref{assum:cvg} hold. It holds that 
		\[\mathbb{P}(\|\mathbf{x}(\infty)-\bar{\mathbf{x}}(\infty)\|\leq \epsilon_n'\|\mathbf{x}(0)\|)> 1-O(n^{-q}), \]
		with 
		\[\epsilon_n'=\frac{6\sqrt{\Delta\log n}}{\delta}+\frac{4(2-\theta)\sqrt{\Delta\log n} \Delta}{\theta\delta^2} \]
		and $q=\min\{\frac{c_2}{8}-1,\frac{1}{5} \}$. 
	\end{theorem}
	
	\begin{proof}
		As $\V_s=\V$, \eqref{eq:dif} becomes  
		\begin{equation}\label{eq:dif_all_stub}
			\begin{aligned}
				\|\mathbf{x}(\infty)-\bar{\mathbf{x}}(\infty)\|&\leq \frac{\theta}{1-\theta}\|M^{-1}\|(\|D-\Bar{D}\|\\
				&+\|\Bar{M}^{-1}\|\|M-\Bar{M}\|\|\Bar{D}\|)\|\mathbf{x}(0)\|.
			\end{aligned}
		\end{equation}
		From the Gershgorin disk theorem, we have $\lambda_{\min}(M)\geq \frac{\theta}{1-\theta}d_{\min}$. Following the same argument as in Lemma \ref{le:degree}, we have $\mathbb{P}(d_{\min}\leq \delta/2)\leq ne^{-\delta/8}$. Therefore, it holds that
		\begin{equation}
			\mathbb{P}(\|M^{-1}\|\leq \frac{2(1-\theta)}{\theta\delta})>1- ne^{-\delta/8}.
		\end{equation}
		On the other hand, from \eqref{eq:D_dif}, if $\Delta>\log n$, we can obtain
		\begin{equation}
			\mathbb{P}(\|D-\bar{D}\|\geq 3\sqrt{\Delta\log n}) \leq 2n^{-\frac{1}{2}}.
		\end{equation}
		Similarly to the proof of Theorem \ref{th:concen}, for any $a>0$, we have
		\begin{equation*}
			\mathbb{P}(\|M-\bar{M}\|\geq a)\leq 2n\exp\left\{\frac{-a^2/2}{o^2+(1+\frac{1}{1-\theta})a/3} \right\},
		\end{equation*}
		with $o^2=(1+\frac{1}{1-\theta})^2\Delta$. Taking $a=2(1+\frac{1}{1-\theta})\sqrt{\Delta\log n}$, it then holds
		\begin{equation}\label{eq:M_dif_all_stub}
			\begin{aligned}
				\mathbb{P}(\|M-\bar{M}\|\geq 2(1+\frac{1}{1-\theta})\sqrt{\Delta\log n})\leq 2n^{-\frac{1}{5}},
			\end{aligned}
		\end{equation}
		Combining \eqref{eq:dif_all_stub}-\eqref{eq:M_dif_all_stub}, we can conclude that with probability at least of $1-ne^{-\delta}-2n^{-\frac{1}{2}}-2n^{-\frac{1}{5}}$, it holds
		\begin{equation*}
			\begin{aligned}
				&\|\mathbf{x}(\infty)-\bar{\mathbf{x}}(\infty)\|<\frac{\theta}{1-\theta}\cdot \frac{2(1-\theta)}{\theta\delta} \Big(3\sqrt{\Delta\log n}\\
				& +\frac{1-\theta}{\theta\delta} \cdot 2(1+\frac{1}{1-\theta})\sqrt{\Delta\log n} \Delta\Big)\|\mathbf{x}(0)\|=\epsilon_n'\|\mathbf{x}(0)\|.
			\end{aligned}
		\end{equation*}
		According to Assumption \ref{assum:bound}, it holds $\delta >c_2\log n>8\log n$, the desired conclusion is then obtained.
	\end{proof}
	
	\begin{remark}\label{re:complex_stub}
		According to Theorem \ref{th:all_stubborn}, if all agents are stubborn, the probability bound $\epsilon_n'$ is monotonically decreasing with $\theta$. This is intuitive, since stubbornness represents memory of each agent toward its initial opinion, which is deterministic and can be regarded as an ``anchor" against the randomness of social interactions in the opinion evolution process.
	\end{remark}
	
	From Theorem \ref{th:all_stubborn}, as $\theta$ goes close to $1$, the probability bound $\epsilon_n'$ tends to be $\frac{6\sqrt{\Delta\log n}}{\delta}+\frac{4\sqrt{\Delta\log n} \Delta}{\delta^2}$, instead of~$0$. This means that $\epsilon_n'$ is still a conservative bound, whose conservativeness is caused by the inequalities of matrix norms in \eqref{eq:es_dif} and \eqref{eq:dif}. To obtain a more accurate bound is a work for future studies.
	
	% \begin{remark}\label{re:bd_nstub_degree}
		%     Even if the probability bounds given in Theorem~\ref{th:concen} and \ref{th:all_stubborn} are not compact, they can still provide useful information on the concentration property. At first, from \eqref{eq:prob_bd} we know that $\epsilon_n$ is monotonically decreasing w.r.t $\delta^s$ and $\delta^{rs}$, which implies that increasing the connectivity among stubborn agents and from non-stubborn agents to stubborn agents can reduce the distance between the final opinions and the expected ones. Moreover, \eqref{eq:prob_bd} also indicates that the connectivity among non-stubborn agents have less impact on the probability bound, which will be verified in the following Example \ref{ex:bd_degree}. 
		% \end{remark}

	\section{Examples}\label{sec:example}
	
	In this section, we apply the results in Section \ref{sec:result} to an SBM with communities $\V_s$ and $\V_r$, defined as follows.   
	
	Suppose the link probability matrix between the communities is 
	\[\Pi=\left(\begin{array}{cc}
		p_s & p_{sr} \\
		p_{sr} & p_r
	\end{array}\right), \]
	where $p_s,p_r, p_{sr}\in (0,1)$ are constants. Let $r_s$ be the ratio of stubborn agents in the whole group, i.e., $r_s=\frac{n_s}{n}$, and correspondingly $1-r_s$ is the ratio of non-stubborn agents. We use the pair $(r_s, \Pi)$ to represent the SBM. The following corollary is obtained by applying Theorem~\ref{th:concen}.
	
	\begin{corollary}\label{coro:concen}
		Consider the model \eqref{eq:FJ_rand} on the SBM $(r_s,\Pi)$, with agent stubbornness defined in \eqref{eq:stb}. Assume that $r_s\in (0,1)$. If $n$ is large enough, it holds that 
		\[\mathbb{P}(\|\mathbf{x}(\infty)-\bar{\mathbf{x}}(\infty)\|\leq \bar{\epsilon}_n\|\mathbf{x}(0)\|)> 1-O(n^{-\frac{1}{5}}), \]
		or equivalently, $\mathbb{P}(\|P-\bar{P}\|\leq \bar{\epsilon}_n)> 1-O(n^{-\frac{1}{5}})$,
		with
		\begin{equation}\label{eq:bd_SBM}
			\begin{aligned}
				\bar{\epsilon}_n&= \frac{6}{1-\theta}\cdot\left(\frac{\sqrt{b_2}}{b_1}+\frac{1-\theta}{\sqrt{b_2}}\right)\sqrt{\frac{\log n}{n}}+\frac{4(2-\theta)}{\theta(1-\theta)^2}\cdot \\
				&\left(\frac{\sqrt{b_2}}{b_1}+\frac{1-\theta}{\sqrt{b_2}}\right)^2\cdot
				\sqrt{\max\{b_2,b_3\}}\sqrt{\frac{\log n}{n}}
			\end{aligned}
		\end{equation}
		where $b_1:=r_sp_{sr}$, $b_2:=r_sp_s+(1-r_s)p_{sr}$ and $b_3:=r_sp_{sr}+(1-r_s)p_r$.
	\end{corollary}
	
	\begin{remark}\label{re:bound_net_size}
		Corollary \ref{coro:concen} implies that as $n$ grows, the probability bound $\bar{\epsilon}_n$ is $O(\sqrt{\frac{\log n}{n}})$. As shown in the following Example \ref{ex:bd_net_size}, this order is close to that of the compact bound. 
	\end{remark}
	
	% \begin{remark}\label{re:bd_degree}
		%     From \eqref{eq:bd_SBM}, if $r_s\leq \frac{1}{2}$, $\bar{\epsilon}_n$ is monotonically decreasing with $p_{sr}$. This indicates that the more densely the non-stubborn community is connected to the stubborn community, the smaller the probability bound will be. This is validated in the following Example \ref{ex:bd_degree}. 
		% \end{remark}
	
	\begin{figure}[ht]
		\centering
		\includegraphics[trim=0cm 0cm 0cm 0cm,clip=true,width=8cm]{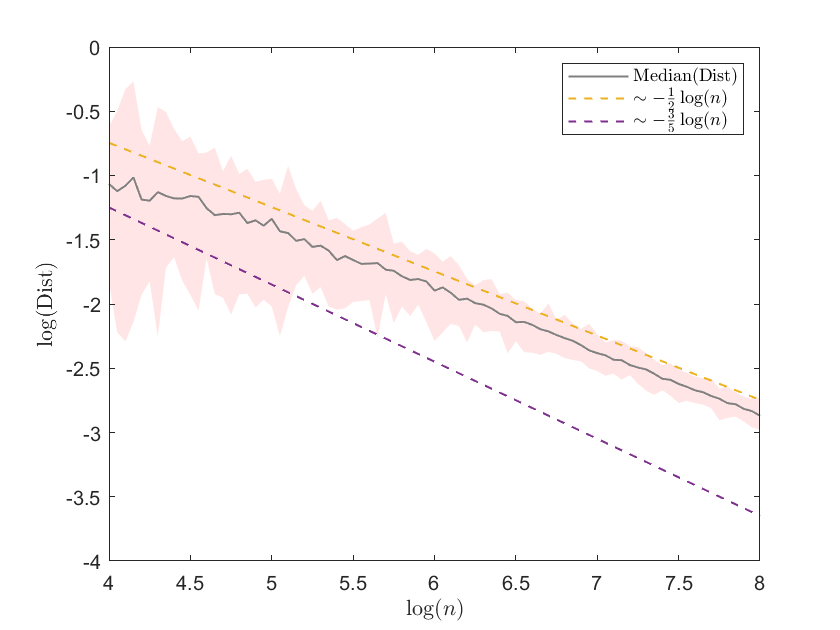}
		\caption{The $\log$-$\log$ plot of $\mathrm{Dist}$ w.r.t the network size $n$ for Example \ref{ex:bd_net_size}. The red shade represents the range of all $\mathrm{Dist}$s generated in the simulation. The gray curve is the change of the medians of $\mathrm{Dist}$ for each $n$. The two dashed curves are drawn for comparison, with slopes $-\frac{1}{2}$ and $-\frac{3}{5}$, respectively.  }
		\label{fig:bd_net_size}
	\end{figure}
	
	\begin{example}\label{ex:bd_net_size}
		Consider the SBM $(r_s,\Pi)$ with $r_s=0.1, p_s=0.3, p_r=0.3$ and $p_{sr}=0.5$. Let the stubbornness be $\theta=0.5$. For each $n$ in $\{[e^4],[e^{4.05}],\dots,[e^8]\}$, we run the SBM $(r_s,\Pi)$ for $50$ times and obtain $50$ graphs of size $n$. For each of the graphs, we calculate $\mathrm{Dist}=\|P-\bar{P}\|$, of which the $\epsilon_n$ in Theorem \ref{th:concen} is a probability bound. The results are shown in Fig. \ref{fig:bd_net_size}. The orange dashed line has slope $-\frac{1}{2}$, and the red shade is the range of all $\mathrm{Dist}s$ generated in the simulation. It can be seen that as $n$ grows, the gray curve tend to be parallel to the orange dashed line. This means that the order of $\mathrm{Dist}$ is $O(n^{-\frac{1}{2}})$, which is close to that given by Remark~\ref{re:bound_net_size}. Moreover, note that the red shade becomes narrower for larger $n$. This indicates that the opinion distance $\mathrm{Dist}$ concentrates around the gray curve as $n$ grows. Therefore, the simulation results match the conclusions of Theorem \ref{th:concen} and Corollary~\ref{coro:concen}. 
	\end{example}
	
	\begin{figure}[ht]
		\centering
		\includegraphics[trim=0cm 0cm 0cm 0cm,clip=true,width=8cm]{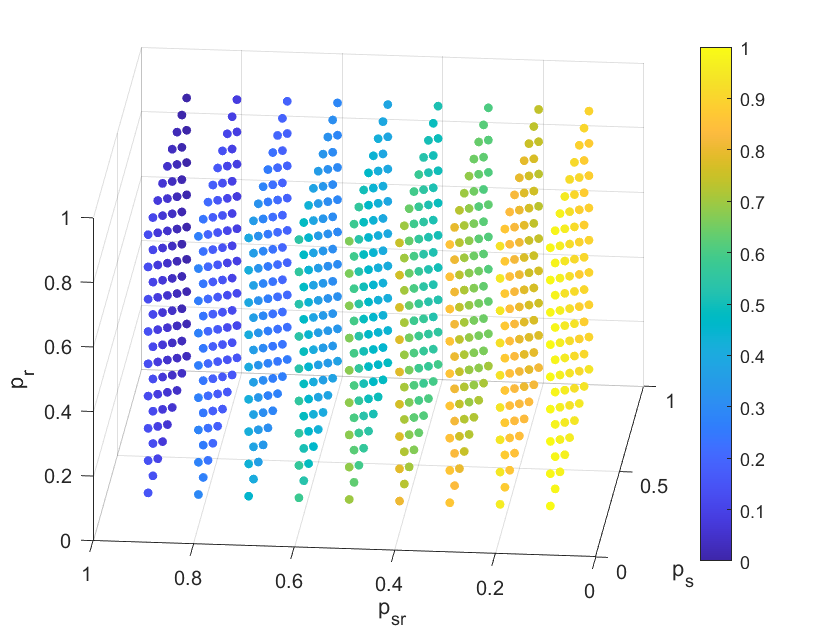}
		\caption{The scatter plot of the medians w.r.t. the degree parameter triplet $(p_s,p_r,p_{sr})$ for Example \ref{ex:bd_degree}. The node color is increasing with the corresponding median. }
		\label{fig:bd_degree}
	\end{figure}
	
	\begin{example}\label{ex:bd_degree}
		Consider the SBM $(r_s,\Pi)$ with $n=500$ and $r_s=0.5$. The stubbornness of stubborn agents are $\theta=0.5$. Let the triplet of degree parameters $(p_s,p_r,p_{sr})$ varies and takes value from the set $\{0.1,0.2,\dots,0.9\}^3$. For each triplet $(p_s,p_r,p_{sr})$, the SBM $(r_s,\Pi)$ is run for $50$ times, and generates $50$ graphs. For each of the graphs, we calculate $\mathrm{Dist}$ (the definition is the same as that of Example \ref{ex:bd_net_size}), and further obtain the median of the $50$ $\mathrm{Dist}$s. In total, $9^3$ medians are collected, each corresponding to a triplet $(p_s,p_r,p_{sr})$. In Figure \ref{fig:bd_degree}, we use scatter plot to show these medians. Each node is associated with a triplet, with its color increasing with the corresponding median. As shown in the figure, a clear observation is that the median is smaller when $(p_s,p_r,p_{sr})$ is closer to the vertex $(1,1,1)$, which indicates that the opinion distance $\mathrm{Dist}$ will be smaller (with high probability) if the underlying graph of the SBM $(r_s,\Pi)$ is more densely connected (i.e., the entries of $\Pi$ are larger). Moreover, from Figure \ref{fig:bd_degree}, the dependency of the median on $p_{sr}$ seems more obvious than that of $p_s$ and $p_r$, which gives evidence to Remark \ref{re:bd_nstub_degree}. 
	\end{example}
	
	\begin{figure}[ht]
		\centering
		\includegraphics[trim=0cm 0cm 0cm 0cm,clip=true,width=8cm]{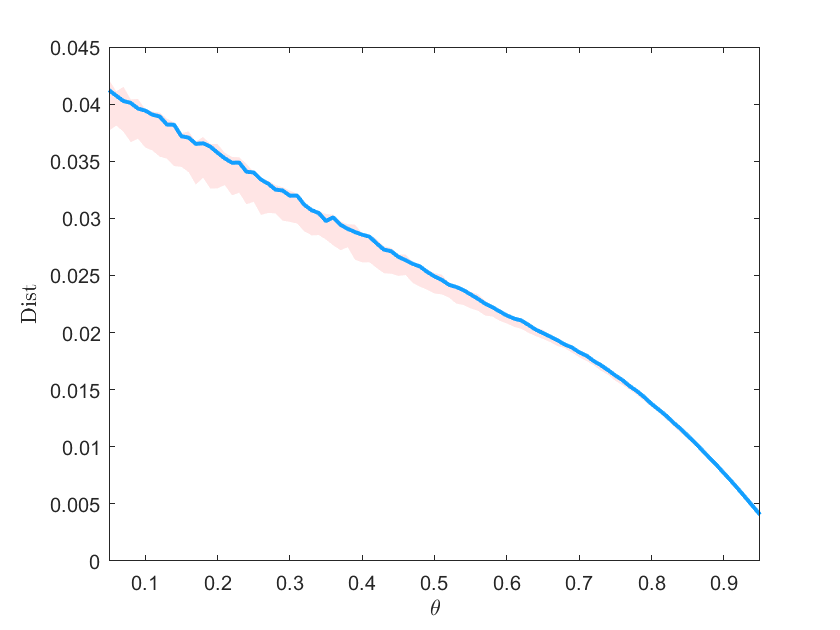}
		\caption{Plot of $\mathrm{Dist}$ w.r.t. the agent stubbornness $\theta$ for Example \ref{ex:bd_stub}. The red shade represents the range of all $\mathrm{Dist}$s generated in the simulation.}
		\label{fig:bd_stub}
	\end{figure}
	
	Lastly, we study an FJ model with all the agents being stubborn.
	
	\begin{example}\label{ex:bd_stub}
		Consider the SBM $(r_s,\Pi)$ with $n=1000$ agents. The degree parameters are set as $p_s=p_r=p_{sr}=0.2$. All of the agents are stubborn, i.e., $r_s=1$. Let the agent stubbornness $\theta$ vary and take values from the set $\{0.05,0.06,\dots, 0.95\}$. For each $\theta$, the SBM is run for $50$ times and generates $50$ graphs. For each of the graphs, we calculate $\mathrm{Dist}$ (defined the same as that of the previous examples). The results are shown in Figure \ref{fig:bd_stub}. In the figure, the blue curve represents the $0.95$-quantile for each $\theta$, and the red shade is the range of all $\mathrm{Dist}$s. As indicated by Remark~\ref{re:complex_stub}, the probability bound $\mathrm{Dist}$ is monotonically decreasing with the stubbornness $\theta$.   
	\end{example}

	\section{Conclusion}\label{sec:conclusion}
	
	This paper considered the FJ model over a random graph. The stubbornness of the stubborn agents were homogenous. The graph was realized by a random graph model, with each edge added from a Bernoulli distribution, and for different edges the distributions are independent. It has been shown that, if the non-stubborn agents were well-connected to the stubborn agents in the expected graph of the random graph model, the final opinions of the FJ model over any realized graph would concentrate around those of the FJ model over the expected graph. Moreover, as indicated by both the theoretical results and numerical examples, the opinion distance was affected significantly by the link probabilities from non-stubborn agents to stubborn agents. For the case that all of the agents were stubborn, the opinion distance would decrease with the agent stubbornness. As mentioned in the context, the proposed probability bounds for the opinion distance were conservative, which left a future direction to seek for more compact bounds.

	\bibliographystyle{abbrv}
	\bibliography{ref_concentration_FJ}
\end{document}